\documentclass{ws-acs}
\newtheorem{problem}{Problem}

\usepackage{graphicx}
\usepackage{amsmath,amssymb,amsfonts}
\usepackage{mathrsfs}
\usepackage{xcolor}
\usepackage{textcomp}
\usepackage{booktabs}
\usepackage{multirow}
\usepackage{url}

\usepackage{algorithm}
\usepackage{algorithmicx}
\usepackage{algpseudocode}

\usepackage{tikz}
\usetikzlibrary{decorations.pathreplacing,positioning,arrows.meta,shapes,patterns}
\usetikzlibrary{calc,fit}
\tikzstyle{vertex}=[draw,circle,fill=white,inner sep=2pt] %
\usepackage{pgfplots}
\usepackage{pgfplotstable}
\pgfplotsset{compat=1.18}

\usepackage[compress]{cite}
\usepackage[verbose]{hyperref}
\hypersetup{colorlinks=false,allbordercolors=blue,pdfborderstyle={/S/U/W 1}}

\newcommand{\orcid}[1]{\href{https://orcid.org/#1}{ORCID: #1}}

\newcommand{\TDom}{{\mathcal{T}}}

\newcommand{\SetCover}{{\sf Min Set Cover}}

\newcommand{\VertexSep}{{\sf d-VertexSep}}
\newcommand{\TimeSep}{{\sf d-MinIntSep}}

\begin{document}

\makeatletter
\def\@biblabel#1{[#1]}
\makeatother

\markboth{R. Dondi and M. M. Hosseinzadeh}{Testing Robustness of Temporal Transportation Networks via Interval Separators}


\title{Testing Robustness of Temporal Transportation Networks via Interval Separators\footnote{A preliminary version of this work appeared in the proceedings of COMPLEX NETWORKS 2025 \cite{DondiHosseinzadeh2025Interval}.}}

\author{Riccardo Dondi\footnote{Universit\`a degli Studi di Bergamo, Bergamo, Italy.}~\orcid{0000-0002-6124-2965}}
\address{Universit\`a degli Studi di Bergamo, Bergamo, Italy\\
\email{riccardo.dondi@unibg.it}}

\author{Mohammad Mehdi Hosseinzadeh\footnotemark[1]~\orcid{0000-0003-3275-6286}}
\address{Universit\`a degli Studi di Bergamo, Bergamo, Italy\\
\email{m.hosseinzadeh@unibg.it}}

\maketitle

\begin{abstract}
This paper addresses the problem of identifying time interval separators in temporal networks. We introduce \TimeSep{}, a new variant of the temporal separator problem, which models failures as time intervals assigned to vertices and aims to block all temporal paths between a source and a target that can be completed within a given deadline $d$. We prove that the \TimeSep{} problem is NP-hard and hard to approximate within a logarithmic function of the size of the vertex set, assuming 
P $\neq$ NP, and we propose an Integer Linear Programming (ILP) formulation to compute minimum interval separators. This latter method is evaluated on synthetic and real-world temporal networks derived from transportation datasets. The experimental results show that the running time is strongly influenced by the temporal dimension, the imposed deadline, and the density of temporal paths.
\end{abstract}

\keywords{Temporal graphs; Temporal separator problem; Integer Linear Programming; Computational complexity; Algorithms for Network Analysis.}

\section{Introduction}
\label{sec:intro}

Novel representations of entity interactions have been considered
in network science and graph literature, in order to take into account
their dynamics and heterogeneity. This has led to the definition
of novel graph models, a notable example being \emph{temporal networks}
\cite{holme2015modern,holme2019map,DBLP:journals/netmahib/HosseinzadehCGD23,DBLP:journals/im/Michail16}.

Temporal networks or temporal graphs
represent how interactions (or arcs) evolve in
a discrete time domain
for a given set of entities (or vertices)~\cite{holme2015modern,DBLP:journals/jcss/KempeKK02}.
The time domain consists of a sequence of timestamps associated with arcs.
A temporal graph can also be seen as a sequence of static graphs,
one for each timestamp, over the same vertex set,
while the arc sets can change from one timestamp to the other.
In particular, an arc observed in a timestamp is called a \emph{temporal arc}.


Research on temporal graphs has focused on studying connectivity and finding paths or walks. Early work introduced the basic models and reachability concepts
\cite{DBLP:journals/jcss/KempeKK02,DBLP:journals/tcs/FluschnikMNRZ20}. Subsequent research has investigated temporal connectivity from an algorithmic perspective, 
addressing both computational complexity and the development of efficient solution methods
\cite{DBLP:journals/jcss/AkridaMSR21,DBLP:conf/iwoca/BumpusM21,DBLP:journals/jcss/Erlebach0K21,DBLP:journals/jcss/ZschocheFMN20,DBLP:journals/algorithmica/MarinoS23,DBLP:journals/jcss/CostaLMS24}. 
In parallel, temporal connectivity has been explored in empirical analysis and real-world applications
\cite{DBLP:journals/pvldb/WuCHKLX14,DBLP:journals/tkde/WuCKHHW16,dondi2023colorful}. Moreover, other problems have been studied, for example dense subgraph discovery \cite{DBLP:journals/sncs/DondiH21,DBLP:journals/kais/RozenshteinBGST20} or temporal covering \cite{DBLP:conf/aaai/HammKMS22,DBLP:journals/datamine/RozenshteinTG21}.

In this paper, we study an approach proposed in \cite{DBLP:conf/isaac/HarutyunyanKP23}
that considers temporal graphs to test the robustness of a transportation system against failures of some elements.
A classical robustness measure is the size of a minimum vertex separator: a minimum set of vertices whose removal separates all paths between a source and a target.
In temporal graphs, time labels lead to temporal paths, i.e., paths that respect time: a temporal arc must be followed by a temporal arc with an increasing time label (strict model).
A temporal vertex separator is then defined as a minimal set of vertices whose removal separates all temporal paths from the source to the target \cite{DBLP:journals/jcss/ZschocheFMN20}.

The model in \cite{DBLP:conf/isaac/HarutyunyanKP23} introduces \VertexSep{}, where, given a temporal graph and a parameter $d$, one seeks a set of vertices whose deletion removes all temporal paths that take time at most $d$.
The problem is NP-hard \cite{DBLP:conf/isaac/HarutyunyanKP23}, and its parameterized and approximation complexity has been further analyzed \cite{DBLP:conf/wads/DondiL25,DBLP:conf/isaac/HarutyunyanKP23}.

Here, we follow this approach and further investigate separators depending on time, motivated by temporal failures of vertices.
Even limited-time failures may affect connectivity for a long time (e.g., a temporally failed train station).
We introduce the concept of separator timeline, inspired by \cite{DBLP:journals/jcss/DondiL25,DBLP:journals/datamine/RozenshteinTG21}.
A separator timeline is a set of time intervals associated with vertices so that temporal paths from the source to the target are deleted.
More precisely, we define \TimeSep{}, which asks for a minimum-length set of vertex intervals intersecting each temporal path between source and target that can be completed within deadline $d$ (see the example in Fig.~\ref{fig:example}).
We prove that \TimeSep{} is not only NP-hard, but also
hard to approximate within factor  $a \ln |V| $, for any constant $a$ with $0 < a < 1$ ($V$ is the vertex set),
assuming P $\neq $ NP (Section~\ref{sec:Complexity}).
Then in Section~\ref{sec:ILP} we design an ILP method, and we evaluate it experimentally (Section~\ref{sec:Experiments}).
We conclude in Section~\ref{sec:Conclusion} pointing out some future directions.

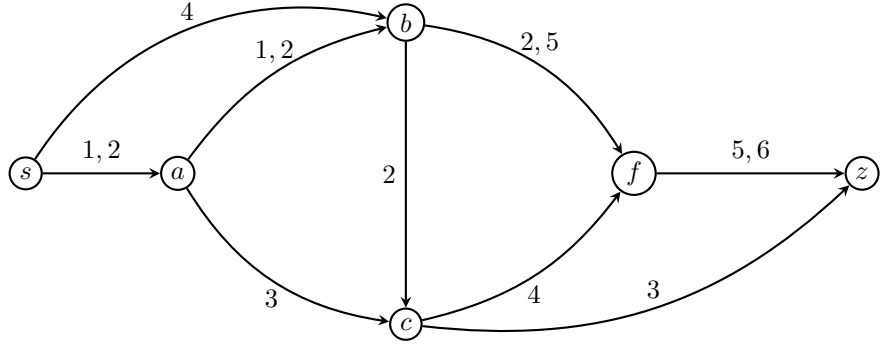
\begin{figure}
    \centering

  \begin{tikzpicture}[->,>=stealth,thick]
    \coordinate (a) at (0,0);
    \coordinate (b) at (3,2);
    \coordinate (c) at (3,-2);
    \coordinate (d) at (6,0);

    \node[vertex] (S) at (-2,0) {$s$};
    \node[vertex] (A) at (a) {$a$};
    \node[vertex] (B) at (b) {$b$};
    \node[vertex] (C) at (c) {$c$};
    \node[vertex] (D) at (d) {$f$};
    \node[vertex] (Z) at (9,0) {$z$};

\draw (S) to node[above] 
{$1 , 2$} (A);

\draw[bend left=35] (S) to node[above] 
{$4$} (B);

\draw[bend left=20] (A) to node[above] 
{$1,2$} (B);

\draw (B) to node[left] 
{$2$} (C);

\draw[bend right=25] (A) to node[below] 
{$3$} (C);

\draw[bend left=25] (B) to node[above] 
{$2, 5$} (D);

\draw[bend right=20] (C) to node[below] 
{$4$} (D);

\draw (D) to node[above] 
{$5, 6$} (Z);

\draw[bend right=25] (C) to node[above] 
{$3$} (Z);

  \end{tikzpicture}
    \caption{An example of a temporal graph $\mathcal{D}$. Consider the case that the deadline $d = 4$, then there exists two $s$-$z$ temporal paths in $\mathcal{D}$
    that are within the deadline:
    $s, (sb, 4), b, (bf,5), f, (fz, 6), z$
    and $s, (sa, 2), a, (ac,3), c, (cf, 4), f, (fz,5), z$.
    Note that other $s$-$z$ temporal paths have traveling time not within the deadline, for example,
    $s, (sa, 1), a, (ab,2), b , (bf,5), f, (fz, 6), z$
    has traveling time $6-2 + 1= 5 > d = 4$.
    An $s,z,d$-separator of $D$ is $I_f = [5,6]$; it has length $2$.
    }
    \label{fig:example}
\end{figure}

\section{Preliminaries}
\label{sec:prelim}

A temporal graph is defined over timestamps $1$ to $T$, denoted by $[1,\dots,T]$.
For an integer $n$, $[n]$ denotes $[1,\dots,n]$.
An interval $I=[i,j]$ with $1\le i\le j\le T$ has length $l(I)=j-i+1$.
An empty interval has length $0$.

\begin{definition}
\label{def:TempGraph}
A temporal directed graph $\mathcal{D}=(D=(V,A),\lambda)$ consists of (1) an underlying directed graph $D=(V,A)$, and (2) a function $\lambda$ that assigns a set of timestamps to each arc $a\in A$.
\end{definition}

The set of temporal arcs of $\mathcal{D}$ is
\[
A_T =\{ (uv,t): uv \in A \wedge t \in \lambda(uv)\}.
\]

Given a vertex $v\in V$, an interval associated with $v$ is $I_v=[l_v,r_v]$, $1\le l_v\le r_v\le T$.
A temporal path is a sequence of alternating vertices and temporal arcs such that (1) each vertex appears only once and (2) arcs are traversed time-consistently;
since we consider the strict model, for two consecutive arcs $(uv,t)$ and $(vw,t')$, we require $t<t'$.

Given a vertex $v$ and interval $I_v$, we say that $(v,I_v)$ \emph{separates} each temporal path that traverses an outgoing temporal arc $(vu,t)$ with $t\in I_v$.
An empty interval does not separate any temporal path.

Given two vertices $s,z\in V$, an $s,z$-\emph{separator timeline} $\mathcal{S}$ in $\mathcal{D}$ is a set of vertex intervals such that:
\begin{itemlist}
\item $\mathcal{S}$ contains one interval for each $v\in V$;
\item each temporal path from $s$ to $z$ is separated by some interval $I_v\in \mathcal{S}$.
\end{itemlist}
$\mathcal{S}$ may contain empty intervals; $s$ and $z$ have empty intervals in any separator timeline.
The length of $\mathcal{S}$ is
\[
l(\mathcal{S}) = \sum_{v \in V} l(I_v).
\]

Consider a temporal path
\[
P = (u_1,(u_1u_2,t_1),u_2,\dots,(u_{h-1}u_h,t_{h-1}),u_h).
\]
Its traveling time is $trt(P)=t_{h-1}-t_1+1$.
A deadline $d\in[1,T]$ bounds travel time: $trt(P)\le d$.

An $s,z,d$-\emph{separator timeline} $\mathcal{S}$ requires separation only for temporal paths from $s$ to $z$ with traveling time at most $d$ (see Fig.~\ref{fig:example}).

Now, we define the problem we are interested in.

\begin{problem}(\TimeSep{})\\
\textbf{Input:} A temporal directed graph $\mathcal{D}$, two vertices $s$ and $z$, and a deadline $d\in[1,T]$.\\
\textbf{Output:} An $s,z,d$-separator timeline of minimum length in $\mathcal{D}$.
\end{problem}

\section{Hardness}
\label{sec:Complexity}

In this section we prove NP-hardness and hardness of approximation of \TimeSep{} via a reduction from \SetCover{}.
We recall here the definition of \SetCover{}.

\begin{problem}(\SetCover{})\\
\textbf{Input:} A collection $\mathcal{C} = \{ C_1, \dots, C_m\}$ of sets over a universe set $U=\{ u_1, \dots, u_n\}$.\\
\textbf{Output:} A collection $\mathcal{C}' \subseteq \mathcal{C}$  of minimum cardinality such that for each $u_i \in U$, $i \in [n]$, there exists a set $C_j \in \mathcal{C}'$,
$j \in [m]$, with $u_i \in C_j$.
\end{problem}

Given an instance $(\mathcal{C},U)$ of \SetCover{}, we define in the following a corresponding instance $(\mathcal{D},s,z,d)$ of \TimeSep{}
(see the example of Fig.~\ref{fig:reduction}).
Note that we assume that
the sets $C_j$, $j \in [m]$ to which an element $u_i$, $i \in [n]$, belongs to are ordered, that
is if $u_i \in C_a, C_b$, $a,b \in [m]$,
then $C_a$ precedes $C_b$ if $a < b$.
Based on this order, two sets $C_a, C_b \in \mathcal{C}$ are \emph{consecutive} related to $u_i$
if $u_i \in C_a, C_b$ and 
there is no set $y$ with $u_i \in C_y$ and $a < y < b$.
The first (last, respectively) set  of $\mathcal{C}$ containing $u_i$, $i \in [n]$, is the set of $\mathcal{C}$ having minimum  index (maximum index, respectively) containing $u_i$.

Given $(\mathcal{C},U)$ we construct a temporal
graph $\mathcal{D} = (V, A_T)$ (see Fig. \ref{fig:reduction}
for an example of the temporal graph built by the reduction).
We start by the defining the vertex set $V$:
\[
V=\{s,z\}\cup\{v_j: C_j \in \mathcal{C}, j \in [m] \}.
\]

The set $A_T$ of temporal arcs is defined as follows:

\begin{align*}
A_T = & \{(sv_j,2j-1): j \in[m]\}\cup \{(v_j z,2j): j\in[m]\}\cup\\
& \{(s v_j,(mn)^3+2nj+ i-1): C_j \text{ is the first set containing $u_i$}, j \in [m], i\in[n] \}\cup\\ 
& \{(v_iv_y,(mn)^3+2nj+i): \text{ $C_j$ and $C_y$ are consecutive sets related to $u_i$}, i\in [n],\\ 
& j,y\in[m], j < y\} \cup\\
& \{(v_i z,(mn)^3+2nj+i): 
 \text{ $C_j$ is the last sets containing $u_i$}, 
j\in [m], i\in[n]\}.
\end{align*}

Finally, we define the deadline $d=n$.

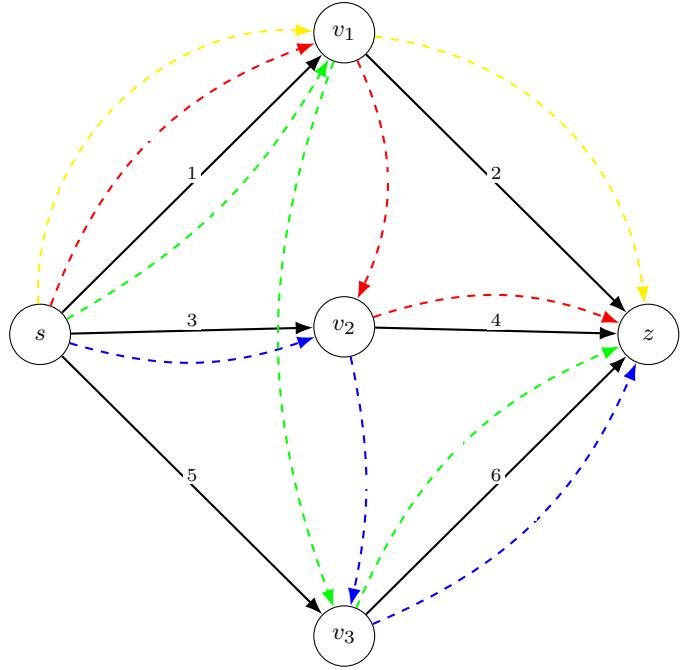
\begin{figure}
    \centering

\begin{tikzpicture}[
  vtx/.style={circle, draw, fill=white, minimum size=8mm, inner sep=0pt, font=\small},
  lab/.style={font=\scriptsize, fill=white, inner sep=1.2pt},
  e/.style={-Latex, thick},
  short/.style={e},
  long/.style={e, dashed},
  box/.style={draw, rounded corners, inner sep=5pt, font=\scriptsize}
]


\node[vtx] (s) at (0,0) {$s$};
\node[vtx] (z) at (8,0) {$z$};

\node[vtx] (v1) at (4, 4.0) {$v_1$};
\node[vtx] (v2) at (4, 0.1) {$v_2$};
\node[vtx] (v3) at (4,-4.0) {$v_3$};

\draw[short] (s) -- node[lab, above] {$1$} (v1);
\draw[short] (s) -- node[lab, above] {$3$} (v2);
\draw[short] (s) -- node[lab, above] {$5$} (v3);

\draw[short] (v1) -- node[lab, above] {$2$} (z);
\draw[short] (v2) -- node[lab, above] {$4$} (z);
\draw[short] (v3) -- node[lab, above] {$6$} (z);


\draw[long, bend left=25, red] (s) to node[lab, above] 
{ } (v1);
\draw[long, bend left=25,red] (v1) to node[lab, above] 
{ }
(v2);
\draw[long, bend left=20,red] (v2) to node[lab, below] 
{ }
(z);

\draw[long, bend right=18, blue] (s) to node[lab, below] 
{ }
(v2);
\draw[long, bend left=12, blue] (v2) to node[lab, below] 
{ }
(v3);
\draw[long, bend right=22, blue] (v3) to node[lab, above] 
{ }
(z);

\draw[long, bend left=48, yellow] (s) to node[lab, below] 
{ }
(v1);
\draw[long, bend left=38, yellow] (v1) to node[lab, below]
{ }
(z);

\draw[long, bend right=15, green] (s) to node[lab, below] 
{ }
(v1);
\draw[long, bend right=20, green] (v1) to node[lab, below] 
{ }
(v3);
\draw[long, bend left=22, green] (v3) to node[lab, above] 
{ }
(z);



\end{tikzpicture}

    \caption{An example of temporal graph built by the reduction, where $U = \{ u_1, u_2, u_3, u_4\}$,
    $C_1 = \{ u_1, u_3, u_4\}$, $C_2 = \{ u_1, u_2\}$
    and $C_3 = \{ u_2, u_4\}$,
    thus $u_1 \in C_1 \cup C_2$ (the temporal path encoding $u_1$
    is dashed red), $u_2 \in C_2 \cup C_3$
    (the temporal path encoding $u_2$ is dashed blue),
    $u_3 \in C_1$ (the temporal path encoding $u_3$ is dashed yellow)
    and $u_4 \in C_1 \cup C_3$
    (the temporal path encoding $u_4$ is dashed green).
    The timestamps of the dashed temporal arcs are not shown
    for readability.
    The temporal arcs on the dashed red path have timestamps
    $(mn)^3 + 8$, $(mn)^3 + 9$ and $(mn)^3 + 10$,
    the temporal arcs on the dashed blue path have timestamps
    $(mn)^3 + 17$, $(mn)^3 + 18$ and $(mn)^3 + 19$,
     the temporal arcs on the dashed yellow path have timestamps $(mn)^3 + 24$, $(mn)^3 + 25$,
    the temporal arcs on the green dashed path have timestamps
    $(mn)^3 + 32$, $(mn)^3 + 33$ and $(mn)^3 + 35$.
    Note that $d= n = 4$ and $(mn)^3 = (12)^3 = 1728$.
    The definition of timestamps implies that the temporal
    paths of traveling time at most $d$ are only those induced
    by temporal arcs having the same color.
    Note that each temporal path $(s, (s v_i,2i -  1), v_i,
    (v_i, z, 2i), z)$, $i \in [m]$, has to be separated
    by an interval $I_{v_i}$ that includes time $2i$.
    }
    \label{fig:reduction}
\end{figure}

We start by proving the following property of 
the temporal paths in $\mathcal{D}$.

\begin{lemma}
\label{lem:prel}
Let $(\mathcal{C},U)$ be an instance of \SetCover{}, and let $(\mathcal{D},s,z,d)$ be the corresponding instance of \TimeSep{}.
Each temporal path of traveling time at most $d$ in $\mathcal{D}$ is either:
\begin{romanlist}[(i)]
\item $(s,(sv_j,2j-1),v_j,(v_ j z,2j),z)$ with $j\in[m]$; or
\item a temporal path containing temporal arcs
defined in interval $[(mn)^3+2nj, (mn)^3+ 2nj + n]$,
with $j \in [m]$ 
\end{romanlist}
\end{lemma}

\begin{proof}
By construction, the first arc of any $s$--$z$ temporal path $p$ in $\mathcal{D}$ is $(sv_j,t)$ where $t=2j-1$ or $t=(mn)^3+2nj+i-1$, for some $j \in [m]$ and some $i \in [n]$.
If $t=2j-1$, then due to deadline $d=n$, $p = (s,(sv_j,2j-1),v_j,(v_j z,2j),z)$ since the other temporal arcs outgoing from $v_i$ are defined at timestamp greater than $(mn)^3 > 2j + d$ .

If $t=(mn)^3+2nj+i-1$, we claim that the temporal
arcs of $p$ are defined in interval $[(mn)^3+2nj, (mn)^3+ 2nj + n]$.
Since the first temporal arc of $p$ is $(s v_j,t)$ with 
$t=(mn)^3+2nj+i-1$, then each temporal arc in $p$
must be defined at time $t'$ with 
$t' > (mn)^3+2nj + i-1 \geq (mn)^3 + 2nj $ for time consistency.
On the other hand, 
$t' < (mn)^3 + 2nj + n + 1$.
Indeed, $t' < (mn)^3 + 2nj + i-1 + n$
since otherwise the traveling time of $p$ would be
larger than $d = n$.
The temporal arcs must then have time 
$t' \leq (mn)^3 + 2nj + n$, since by definition of $\mathcal{D}$ for a fixed $j \in [m]$, the temporal arcs have maximum time $(mn)^3 + 2nj + n$.
On the other hand, the temporal arcs defined for $y > j$
are defined at time at least
$(mn)^3 + 2n(j+1) = (mn)^3 + 2nj + 2n > 
(mn)^3 + 2nj + i-1 + n$.
\end{proof}

Now, we prove one of the main property of the reduction.

\begin{lemma}
\label{lem:hard1}
Let $(\mathcal{C},U)$ be an instance of \SetCover{}, and $(\mathcal{D},s,z,d)$ the corresponding instance of \TimeSep{}.
Given a solution $\mathcal{C}'$ of \SetCover{}, 
we can compute in polynomial time a solution of \TimeSep{} of length at most $((nm)^3+2(mn) + n+1)|\mathcal{C}'| + m - |\mathcal{C}'|$.
\end{lemma}

\begin{proof}
Given a solution $\mathcal{C}'$ of \SetCover{} on instance $(\mathcal{C}, U)$, define a solution $\mathcal{S}$ of \TimeSep{} on instance $(\mathcal{D},s,z,d)$ as follows:
\begin{romanlist}[(i)]
\item For each $C_j\in \mathcal{C}'$, set $I_{v_j}=[2j,(nm)^3+2(mn)+n]$.
\item For each $C_j\notin \mathcal{C}'$, set $I_{v_j}=[2j]$.
\end{romanlist}
Each interval defined at point (i) have length 
$(nm)^3+2(mn) + n -2j +1 < (nm)^3+2(mn) + n +1$; each interval defined at point (ii), has length $1$.
The overall length of $\mathcal{S}$ is then at most
$|\mathcal{C}'| ((nm)^3+2(mn) + n+1) + 
m - |\mathcal{C}'|$.

Now, we show that $\mathcal{S}$ is indeed a s,z,d-separator timeline, that is it intersects every temporal path of travel time at most $d$.
All temporal paths of type (i) in Lemma~\ref{lem:prel} are intersected by $\mathcal{S}$, since $2j\in I_{v_j}$, for each $j \in [m]$.
Consider a temporal path $p$ of type (ii) in Lemma~\ref{lem:prel},
defined in interval $[(mn)^3+2nj, (mn)^3+ 2nj + n]$.
Since $\mathcal{C}'$ is a set cover, 
for each $u_i \in U$, $i \in [n]$, there exists $C_j \in \mathcal{C}'$, with $j \in [m]$,
such that $u_i \in C_j$. By construction it follows that $I_{v_j} = [2j,(nm)^3+2(mn)+n]$
and each temporal arc outgoing  from $v_j$ has timestamps in $I_{v_j}$,
hence also the temporal arc outgoing from $v_j$ in $p$, thus concluding the proof.
\end{proof}

Now, we prove the second main property of the reduction.

\begin{lemma}
\label{lem:hard2}
Let $(\mathcal{C},U)$ be an instance of \SetCover{}, and $(\mathcal{D},s,z,d)$ the corresponding instance of \TimeSep{}.
Given a solution of \TimeSep{} of length 
at least $((nm)^3) h$ and
at most $((nm)^3+2(mn) + n)h + m - h$,
for some $h \geq 1$, we can compute 
in polynomial time a solution of \SetCover{} consisting
of at most $h$ sets.
\end{lemma}

\begin{proof}
Starting from $\mathcal{S}$ of \TimeSep{} of length 
at least $((nm)^3) h$ and at most $((nm)^3+2(mn) + n)h + m - h$, for some $h \geq 1$,
we define a collection $\mathcal{C}'$ of sets in $\mathcal{C}$
as follows:

\begin{itemize}
    \item For each $I_{v_j}$, $j \in [m]$, of length at least $(nm)^3$ in $\mathcal{S}$, add $C_j$ to $\mathcal{C}'$.
\end{itemize}

We claim that $|\mathcal{C}'| \leq h$ and that
$\mathcal{C}'$ is a set cover.

First, assume that $|\mathcal{C}'| > h$, then
there exists at least $h+1$ vertices
with interval of length at least $(mn)^3$ in $\mathcal{S}$.
Then $\mathcal{S}$ has length at least
\[(h+1) (mn)^3 > h((mn)^3+2(mn) + n) + m - h,
\]
since $(mn)^3 > 2(mn)+n + m$ for $m,n \geq 2$.
Hence $|\mathcal{C}'| \leq h$.

Now, we prove that $\mathcal{C}'$ is a set cover.
Assume this is not the case and that
an element $u_i \in U$, $i \in [n]$, is not covered by $\mathcal{C}$. 
Then for each set $C_j$, $j \in [m]$, such that $u_i \in C_j$,
$|I_{v_j}| < (mn)^3$.
Furthermore, each $I_{v_j}$, for a vertex $v_j$ encoding $C_j$, must include time $2j$ in order to separate the temporal path 
$(s,(sv_j,2j-1),v_j,(v_j z,2j),z)$ with $j\in[m]$ (see Lemma \ref{lem:prel}). Since the length of $I_{v_j}$ 
is smaller than $(mn)^3$, then $I_{v_j}$
does not include
any temporal arc outgoing from $v_j$ that is defined at time at least $(mn)^3 +2nj$, for each $j \geq 1$.
Consider the $s$--$z$ temporal path $p$ in $\mathcal{D}$ that traverses the vertices with temporal arcs defined in interval
$[(mn)^3 + 2nj, (mn)^3 +2nj + n]$
and note that $p$ traverses vertices $v_j$, such that $u_i \in C_j$ and has traveling time at most $d$.
Since each $I_{v_j}$ does not include
any temporal arc outgoing from $v_j$ defined at time at least $(mn)^3 +2nj$, $p$
is not separated by $\mathcal{S}$, a contradiction.
Thus $\mathcal{C}'$ is a set cover.

We conclude the proof observing that
any solution of \TimeSep{} has length 
at least $((nm)^3)$ (hence $h \geq 1$), since
$\mathcal{D}$ contains at least one temporal
path that have arcs in interval $[(mn)^3 + 2nj, (mn)^3 +2nj + n]$.
\end{proof}

We can prove now the main result of this section.

\begin{theorem}
\label{teo:teo1}
\TimeSep{} is hard to approximate within factor $a \ln |V|$,
for any constant $a$ with $0 <a < 1$.
\end{theorem}

\begin{proof}
We show that Lemma~\ref{lem:hard1} and Lemma~\ref{lem:hard2} define an approximation preserving reduction from \SetCover{} to \TimeSep{}.

Let $(\mathcal{C},U)$ be an instance of \SetCover{}, and let $(\mathcal{D},s,z,d)$ be the corresponding instance of \TimeSep{}.
Consider an approximated (optimal, respectively) solution $\mathcal{S}_A$ ($\mathcal{S}_O$, respectively) of \TimeSep{}{} on $\mathcal{D}$.
The approximation factor of \TimeSep{} is equal to: $\frac{l(\mathcal{S}_A)}{l(\mathcal{S}_O)}$.

By Lemma \ref{lem:hard2}, starting from $\mathcal{S}_A$ we can compute in polynomial time an approximated solution $\mathcal{C}'_A$ of \SetCover{} on instance $(\mathcal{S}, U)$, such that 
\[
\frac{l(\mathcal{S}_A)}{l(\mathcal{S}_O)} 
\geq \frac{|\mathcal{C}'_A|(mn)^3}{l(\mathcal{S}_O)}.
\]
By Lemma \ref{lem:hard1}, we have that, 
given an optimal solution $\mathcal{C}'_O$ of \SetCover{} on $(\mathcal{S},U)$, it holds that $l(\mathcal{S}_O) \leq |\mathcal{C}'_O|((mn)^3 + 2mn+n) + m -|\mathcal{C}'_O|.$
Thus we obtain that
\[
\frac{l(\mathcal{S}_A)}{l(\mathcal{S}_O)} \geq 
\frac{|\mathcal{C}'_A|(mn)^3}
{l(\mathcal{S}_O)}
\geq 
\frac{|\mathcal{C}'_A|(mn)^3}
{|\mathcal{C}'_O|((mn)^3 + 2mn+n) + m -|\mathcal{C}'_O|}
\geq
\frac{|\mathcal{C}'_A|(mn)^3}
{2|\mathcal{C}'_O|(mn)^3}.
\]
where the last inequality holds since 
$(mn)^3 >  2mn+n+m $ for $m, n \geq 2$.

Since \SetCover{} is not approximable within factor $c \ln |U|$, for any constant $c$ such that $0 < c < 1$, unless $P = NP$~\cite{DBLP:journals/talg/AlonMS06,DBLP:conf/stoc/DinurS14}, then for any constant $b$, with $0 <b<1$, unless $P = NP$, it follows that
 \[
\frac{l(\mathcal{S}_A)}{l(\mathcal{S}_O)} > b \ln |U| \qquad \text{for any constant $0<b<1$}.
 \]
The inapproximability of \SetCover{} within factor $c \ln |U|$, holds also when $|\mathcal{C}| \leq poly(|U|)$ \cite{DBLP:journals/eccc/Nelson07}, thus $|V| = m +2 \leq poly(|U|)$
and $a \ln |V| \leq b \ln |U|$ for a constant $a$,
$0< a <1$,
thus
 \[
\frac{l(\mathcal{S}_A)}{l(\mathcal{S}_O)} > a \ln |V|
\qquad \text{for any constant $0< a <1$},
 \]
 thus concluding the proof.
\end{proof}

Observe that Theorem \ref{teo:teo1} implies also
the NP-hardness of \TimeSep{}.
\section{Method}
\label{sec:ILP}

We approach \TimeSep{} using an Integer Linear Programming (ILP) method.
We aim to find an $s,z,d$-separator timeline $\mathcal{S}$ of minimum length, i.e., vertex intervals $\{I_v\}_{v\in V}$ that separate all temporal paths from $s$ to $z$ with traveling time at most $d$.

\subsection{ILP-Based Separator Computation}

We introduce binary variables $x_{v,t}\in\{0,1\}$ for each $v\in V$ and timestamp $t\in[1,T]$:
\begin{itemlist}
\item $x_{v,t}=1$ iff timestamp $t$ is included in the interval $I_v$;
\item $x_{v,t}=0$ otherwise.
\end{itemlist}

\subsubsection{Objective Function}
\[
\min \sum_{v\in V}\sum_{t=1}^{T} x_{v,t}.
\]

\subsubsection{Path Separation Constraints}
Let $\mathcal{P}_{s,z,d}$ be the set of temporal paths from $s$ to $z$ with traveling time at most $d$.
For each $P_k\in\mathcal{P}_{s,z,d}$ written as
\begin{equation}
P_k=\big(s,(su_1,t_1),u_1,(u_1u_2,t_2),\dots,(u_hz,t_{h+1}),z\big),
\end{equation}
we impose:
\[
\sum_{(vu,t)\ \text{in}\ P_k,\ v\neq s,z} x_{v,t}\ \ge\ 1.
\]
Since enumerating all temporal paths is hard, we generate constraints iteratively by repeatedly finding deadline-respecting temporal paths and adding the corresponding constraint.

\subsubsection{Contiguity Constraints}
To enforce that each $I_v$ is a contiguous interval (no fragmentation), for each vertex $v$ and timestamps $t_1<t_2$:
\[
x_{v,t_1}+x_{v,t_2}-1 \le x_{v,t_1+1}.
\]
This ensures that if two timestamps are selected, intermediate timestamps must also be selected.

\section{Experimental Evaluation}
\label{sec:Experiments}

We now present an experimental evaluation of \TimeSep{} on both synthetic and real networks.
Experiments were run on a MacBook Air with an Apple M1 processor and 8 GB of RAM, using Python 3.13. In our empirical evaluation, we consider a running time of up to one hour as reasonable time for solving an instance.

\subsection{Synthetic Dataset}
We conducted experiments using synthetic temporal graphs derived from static transportation networks\footnote{Transportation Networks for Research Core Team. Transportation Networks for Research. \url{https://github.com/bstabler/TransportationNetworks}. Accessed July 24, 2025.}.
We transformed each static network into a temporal graph via two phases.
We fixed a source $s$ and target $z$ in each network: $s$ is the vertex with maximum out degree, and $z$ the vertex with maximum in degree.

In phase 1, we extracted a shortest path between $s$ and $z$, then removed the internal vertices of this path (except $s$ and $z$) together with all incident arcs, repeating this procedure until $s$ and $z$ became separated.
Moreover, in order to guarantee the existence of several temporal $s$–$z$ paths within the deadline, each arc on an extracted shortest path was assigned multiple timestamps, with 
$4$ - $8$ randomly selected timestamps assigned
to each arc. 
For Berlin-Friedrichshain, this interval was reduced to 4–6, and for ChicagoSketch it was fixed to 2, as larger values resulted in instances that could not be solved within a reasonable time.
The deadline was determined using the number of arcs in the first discovered $s$–$z$ path, multiplied by three; if this value exceeded the maximum timestamp it was capped accordingly, and if it fell below half of the maximum timestamp it was raised to 50\% of the maximum. Notice that, after assigning timestamps to the extracted $s$–$z$ paths, some of these paths may be infeasible as their temporal format can exceed the deadline.

In phase 2, we assigned temporal labels to the remaining static arcs.
Each such arc was given between 2 and 5 (randomly selected) activation timestamps, sampled uniformly from \(\TDom\), resulting in a sparse temporal background in which each arc carries only a small number of temporal labels.



We solved the ILP using Gurobi\footnote{\url{https://www.gurobi.com}}.
Table \ref{tab:network_generation} summarizes the generated temporal graphs. The datasets span a wide range of sizes, from the small Eastern–Massachusetts network to large graphs such as Barcelona, ChicagoSketch, and Munich, which all instances have the same timestamp horizon ($TS = 50$). Moreover, the number and length of extracted shortest $s$-$z$ paths vary considerably across datasets, some instances yield only a few paths with different lengths, while others, such as Barcelona, produce many short alternatives. Deadlines are largely uniform ($d=25$), as they are derived from the first extracted shortest path and bounded by the global horizon, with ChicagoSketch standing out due to its inherently longer routes and a larger deadline. 
 
Additionally, Table \ref{tab:ilp_synthetic_results} reports ILP performance. Separator sizes vary widely across datasets, from small values in Berlin–Prenzlauerberg to larger ones in Barcelona, while this variability, separators consistently involve small set of vertices. Moreover, the average interval lengths remain modest, which shows that effective separation typically relies on blocking short temporal windows rather than long time spans. Solver running time is dominated by the number of feasible temporal paths rather than by network size alone, instances such as ChicagoSketch and Berlin–Friedrichshain, which exhibit millions of temporal paths in Table \ref{tab:network_generation}, lead to the highest running times, while smaller or less temporally dense networks are solved much faster.

\begin{table}[h!]
\small
\caption{Summary of the original static networks and generated temporal networks. 
Vertices and Arcs denote the number of vertices and arcs in the static network; TA is the number of temporal arcs in the generated temporal network; TS is the number of timestamps; $d$ is the deadline; and ``\#shortest paths: \#arcs'' reports, for each instance, the number of extracted shortest $s$--$z$ paths and the corresponding number of arcs in each such path.}
\centering
\begin{tabular}{l|c|c|c|c|c|c}
\toprule
\textbf{Dataset} & Vertices & Arcs & TA & TS & $d$ & $\#$shortest paths: $\#$arcs \\
\colrule
Anaheim & 416 & 914 & 2996 & 50 & 25 & 4 paths: 2, 6, 7, 13 \\
Barcelona & 930 & 2522 & 8549 & 50 & 25 & 9 paths: 2,2,2,2,2,3,4,4,6 \\
Berlin--Friedrichshain & 224 & 523 & 1673 & 50 & 25 & 5 paths: 2,2,3,7,13 \\
Berlin--Prenzlauerberg & 352 & 749 & 2508 & 50 & 25 & 2 paths: 3, 15 \\
ChicagoSketch & 933 & 2950 & 7456 & 50 & 45 & 6 paths: 15,18,19,20,20,26 \\
Eastern--Massachusetts & 74 & 258 & 648 & 50 & 25 & 5 paths: 3,4,5,5,7 \\
Munich & 742 & 1872 & 5964 & 50 & 25 & 4 paths: 3,5,14,29 \\
\botrule
\end{tabular}
\label{tab:network_generation}
\end{table}

\begin{table}[t]
\small
\caption{ILP-based temporal separator results on synthetic transportation networks. 
SL denotes the separator length (number of blocked time--vertex pairs), $|V_{\mathrm{sep}}|$ the number of vertices with non-empty blocking intervals, Avg.~Int.\ the average length of blocking intervals per separator vertex, $\#P$ the number of temporal $s$--$z$ paths within the deadline, and Time the solver running time in seconds.}
\centering
\begin{tabular}{l|c|c|c|c|c}
\toprule
\textbf{Dataset}
& SL
& $|V_{\mathrm{sep}}|$
& Avg.~Int.
& $\#P$
& Time (s) \\
\colrule
Anaheim & 19 & 4 & 4.75 & 996425 & 160 \\
Barcelona & 48 & 9 & 5.33 & 26079 & 25 \\
Berlin--Friedrichshain & 24 & 5 & 4.80 & 4844671 & 1035 \\
Berlin--Prenzlauerberg & 6 & 2 & 3.00 & 1067720 & 207 \\
ChicagoSketch & 10 & 5 & 2.00 & 4489216 & 2082 \\
Eastern--Massachusetts & 27 & 5 & 5.40 & 627543 & 39 \\
Munich & 13 & 3 & 4.33 & 1529491 & 347 \\
\botrule
\end{tabular}
\label{tab:ilp_synthetic_results}
\end{table}


\subsection{Real Dataset}

We used real-world dataset of public transportation networks derived from GTFS (General Transit Feed Specification). GTFS is a standard format used by public transport agencies to publish timetable-based information. We used selected GTFS-derived temporal networks from an existing multi-city public transportation collection \cite{kujala2018collection}\footnote{\url{https://zenodo.org/records/1186215}}. From the original day-long GTFS data, we restricted each temporal network to the first two hours of time in order to reduce the number of timestamps and keep the instances computationally manageable by our method. 

For each network, we initially selected a small number of source $s$ and target $z$ as vertices with high out-degree in the first 10\% of the time window and high in-degree in the last 50\%, respectively, and considered all $s$--$z$ combinations. If no feasible temporal path exists for a given choice, then, the candidate sets are gradually expanded by including additional vertices according to the their degree based ranking until feasible $s$--$z$ paths are found.
Then, $s$--$z$ paths are connected using earliest-arrival temporal paths. Moreover, the deadline for each instance is set to twice the corresponding smallest traveling time. For each city, among the resulting instances, we reported the solution with the largest separator objective value.

Tables~\ref{tab:real_networks} and \ref{tab:ilp_real_results} summarize the experimental results on selected real-world public transportation networks constructed from the first two hours of the original day-long data. Table~\ref{tab:real_networks} reports the corresponding network and temporal characteristics, while Table~\ref{tab:ilp_real_results} presents the ILP-based temporal separator outcomes. Recall that the objective of the ILP is to block all temporal $s$-$z$ paths within the given deadline using a minimum-size temporal separator. As a result, the ILP runtime is influenced by the number of timestamps, the number of temporal $s$–$z$ paths that must be blocked within the deadline, and the size of the resulting separator (i.e., the objective value), rather than by network size alone. 
For example, Berlin, which has the largest network and a relatively large separator, requires the longest runtime, even though the number of feasible temporal paths within the deadline is moderate. 
In contrast, Luxembourg exhibits a much smaller separator but it contains a large number of temporal paths, leading to a substantial number of blocking constraints while remaining solvable within the considered time limits.

\begin{table}[h!]
\small
\caption{Network and temporal characteristics of the real-world transportation networks used in the experiments. 
Vertices denotes the number of vertices, TA the number of temporal arcs, $T$ the number of timestamps, and $d$ the deadline.}
\centering
\begin{tabular}{l|c|c|c|c}
\toprule
\textbf{Dataset}
& Vertices
& TA
& TS
& d \\
\colrule
Berlin        & 3397 & 24369 & 121 & 10 \\
Grenoble     & 892  & 3292  & 117 & 6  \\
Helsinki     & 1292 & 2092  & 121 & 56 \\
Luxembourg   & 1247 & 7904  & 120 & 50 \\
Venice       & 457  & 1111  & 118 & 118 \\
\botrule
\end{tabular}
\label{tab:real_networks}
\end{table}

\begin{table}[h!]
\small
\caption{ILP-based temporal separator results on real-world transportation networks. 
SL denotes the separator length (number of blocked time--vertex pairs), $|V_{\mathrm{sep}}|$ the number of vertices with non-empty blocking intervals, Avg.~Int.\ the average length of blocking intervals per separator vertex, $\#P$ the number of temporal $s$--$z$ paths within the deadline, and Time the solver running time in seconds.}
\centering
\begin{tabular}{l|c|c|c|c|c}
\toprule
\textbf{Dataset}
& SL
& $|V_{\mathrm{sep}}|$
& Avg.~Int.
& $\#P$
& Time (s) \\
\colrule
Berlin        & 60  & 8 & 7.5 & 46     & 879 \\
Grenoble     & 106 & 2 & 53.0 & 71     & 134 \\
Helsinki     & 3   & 3 & 1.0 & 22     & 214 \\
Luxembourg   & 3   & 3 & 1.0 & 2705   & 205 \\
Venice       & 1   & 1 & 1.0 & 16     & 71 \\
\botrule
\end{tabular}
\label{tab:ilp_real_results}
\end{table}

Table~\ref{tab:ilp_real_results} further characterizes the structure of the computed temporal separators on real-world transportation networks. In addition to the separator length, 
the table reports the number of vertices with non-empty blocking intervals, \( |V_{\mathrm{sep}}| \), and the average length of these intervals, 
which together describe how blocking is distributed across vertices. 
In particular, some networks achieve large separators by concentrating blocking on a very small number of vertices over longer contiguous time intervals, as observed in Grenoble, where only two vertices are involved but with very long (on average) blocking intervals. 
In contrast, other networks distribute the separator across a larger set of vertices with shorter blocking intervals, as in Berlin, indicating a more spatially spread but temporally localized intervention strategy. Smaller networks such as Helsinki, 
Luxembourg, and Venice have very compact separators, both in terms of separator length and number of involved vertices, often consisting of single-time interventions on only one or a few vertices.

Overall, these results show that real transportation networks give rise to markedly different temporal separator structures. The running time is influenced not only by network size, but also by the number of temporal \(s\)--\(z\) paths within the deadline.

\section{Conclusion}
\label{sec:Conclusion}

We studied time separation in temporal graphs and introduced \TimeSep{}, which models failures as time intervals on vertices to disconnect $s$ from $z$ within a deadline $d$.
We proved hardness of approximation (and NP-hardness) for \TimeSep{}. Furthermore, we proposed an ILP-based method to compute minimum interval separators of a temporal graph.

Experiments on both synthetic and real temporal transportation networks show that few vertices active over short time intervals are often sufficient to disrupt timely connections, even in highly connected settings. However, runtime grows significantly with larger timestamp ranges, larger deadlines, and higher temporal path density. In particular, the increasing number of feasible temporal paths observed in both synthetic and real networks emerges as the dominant factor affecting computational performance, rather than graph size alone.

Future works include both theoretical and experimental
directions. For the theoretical directions, it would be
interesting to analyze the parameterized complexity of the 
problem, for example fo parameters $T$ and length of an
$s,z,d$-separator.
It would be interesting to study the complexity in the nonstrict model,

\section*{Acknowledgments}

\section*{ORCID}
Riccardo Dondi - \url{https://orcid.org/0000-0002-6124-2965}\\
Mohammad Mehdi Hosseinzadeh - \url{https://orcid.org/0000-0003-3275-6286}

\bibliographystyle{ws-acs}
\bibliography{refs}

\end{document}